\documentclass[10pt, draftclsnofoot, onecolumn]{IEEEtran}
\usepackage{amsmath,amsfonts}
\usepackage{array}
\usepackage{textcomp}
\usepackage{stfloats}
\usepackage{url}
\usepackage{verbatim}
\usepackage{graphicx}
\usepackage{balance}

\usepackage{color}
\usepackage{amsfonts}
\usepackage{amsmath}
\usepackage{amsthm}
\usepackage{amssymb}
\usepackage{alltt}
\usepackage{epsfig}
\usepackage{mathrsfs}
\usepackage{graphicx,ifthen}
\usepackage{epstopdf}
\usepackage{subfigure}
\usepackage{esint}
\usepackage{bm}
\usepackage{enumerate}
\usepackage{cite}
\usepackage{float}

\usepackage{algorithm}
\usepackage{algpseudocode}
\usepackage{booktabs}
\usepackage{array}
\usepackage{ulem}
\usepackage{soul}
\usepackage{verbatim}
\usepackage{hhline}
\usepackage{multirow}
\usepackage{arydshln}
\usepackage{makecell}
\usepackage{booktabs}

\usepackage{textcomp}
\usepackage{xcolor}
\usepackage{multirow}
\usepackage{setspace}
\usepackage{stfloats}

\newtheorem{proposition}{Proposition}

\newcommand{\mb}[1]{\mathbf{#1}}
\newcommand{\mc}[1]{\mathcal{#1}}
\newcommand{\mbb}[1]{\mathbb{#1}}

\newcommand{\mtm}[1]{\mathrm{#1}}

\begin{document}
	\title{Mutual Information-oriented ISAC Beamforming Design for Large Dimensional Antenna Array}
	
	\date{} 
	
	\author{Shanfeng~Xu, 
		Yanshuo~Cheng,
		Siqiang~Wang,
		Xinyi~Wang,~\IEEEmembership{Member,~IEEE,}
		Zhong~Zheng,~\IEEEmembership{Member,~IEEE,}
		Zesong~Fei,~\IEEEmembership{Senior Member,~IEEE}
		\thanks{
			Shanfeng Xu is with the School of Information and Electronics, Beijing
			Institute of Technology, Beijing 100081, China, and also with the  China
			Academy of Electronics and Information Technology, Beijing 100041, China
			(e-mail: xushanfeng88@163.com).
			
			Yanshuo~Cheng,	Siqiang~Wang, 		Xinyi~Wang,		Zhong~Zheng and
			Zesong~Fei are with the School of Information and Electronics, Beijing Institute of Technology, Beijing 100081, China. (e-mails: chengyanshuobit@163.com, 3120205406@bit.edu.cn, bit\_wangxy@163.com, zhong.zheng@bit.edu.cn, feizesong@bit.edu.cn)}
	}

	\maketitle
	
\begin{abstract}
	 In this paper, we study the beamforming design for multiple-input multiple-output (MIMO) ISAC systems, with the weighted mutual information (MI) comprising sensing and communication perspectives adopted as the performance metric. In particular, the weighted sum of the communication mutual information and the sensing mutual information is shown to  asymptotically converge to a deterministic limit when the number of transmit and receive antennas grow to infinity. This deterministic limit is derived by utilizing  the operator-valued free probability theory. Subsequently, an efficient {projected gradient ascent (PGA) algorithm} is proposed to optimize the transmit beamforming matrix with the aim of maximizing the weighted asymptotic MI. Numerical results validate that the derived closed-form expression matches well with the Monte Carlo simulation results and the proposed optimization algorithm is able to improve the weighted asymptotic MI significantly. We also illustrate the trade-off between asymptotic sensing and asymptotic communication MI.
\end{abstract}


\section{Introduction}

The sixth generation (6G) wireless network has stimulated numerous innovative applications, such as smart cities and intelligent transportation. Toward this end, higher requirements of communication and sensing capabilities are raised for numerous nodes within the network. Due to the capability of enabling the dual functionalities of information transmission and target sensing with shared wireless resources, integrated sensing and communication (ISAC) has been regarded as a key enabler for the realization of 6G network \cite{zhang2021enabling}. 

As a promising approach to enhance the performances of ISAC systems, beamforming design has been investigated in numerous works adopting various communication and sensing performance metrics. In \cite{liu2018mu}, the authors considered  transmit
beampattern and signal-to-interference-plus-noise (SINR) as sensing and communication performance metrics respectively.
As a step further, due to its capability of characterizing the lower bound of parameter estimation, the Cram\'{e}r-Rao Bound (CRB) has also been adopted as a sensing performance metric in ISAC systems \cite{wang2022partially}.

{In  most existing works, the performance metrics for sensing and communications are diverse, resulting in difficulties in evaluating the trade-offs between sensing and communication performance. To facilitate  ISAC beamforming design under unified performance metric, the authors in \cite{Al2023KL1} proposed a Kullback-Leibler (KL) divergence based unified performance metric, and analyzed the error rate performance of communication users and detection performance for sensing targets. Based on this unified metric, the authors in \cite{Fei2024KL2} investigated constellation and beamforming design, and depicted the Pareto bound in terms of KL divergence as well as bit error rate and detection probability. Different from KL divergence, which concentrates on the reliability of data transmission, mutual information is closely related to the efficiency of transmission, and has also been utilized as the performance metric in ISAC systems. { In \cite{xu2015radar}, radar MI and  communication channel capacity were formulated for ISAC systems. As a step further, in \cite{liu2019robust}, OFDM-based ISAC waveform is designed with the aim of maximizing the weighted sum of the communication rates and the conditional radar MI.}  However, it should be noted that the aforementioned MI-oriented  works are not conducted within the context of  large dimensional antenna array}, in which case {In the case of large dimensional antenna array, perfect CSI becomes challenging to obtain, thereby affecting the accuracy of MI.}   

{In this paper, we investigate the MI-oriented transmit beamforming design for a MIMO ISAC system, where an ISAC user equipment (UE) senses an extended target  while simultaneously transmitting data to a base station (BS).} Based on the channel state information  (CSI) available at the {UE}, we formulate a transmit beamforming design problem with the aim of maximizing the weighted asymptotic MI.  Applying the operator-valued free probability theory, we derive the closed-form expression of the weighted asymptotic MI and reformulate the beamforming design problem. Subsequently, based on the obtained closed-form expression, we propose an efficient {projected gradient ascent (PGA) algorithm} to solve the problem. Numerical results validate the accuracy of the derived expression, as well as the convergence and effectiveness of the proposed algorithm. In addition, the trade-off between sensing and communication MI is also depicted.

\section{System Model}
\begin{figure}[t]
	\centerline{\includegraphics[width=0.5\columnwidth]{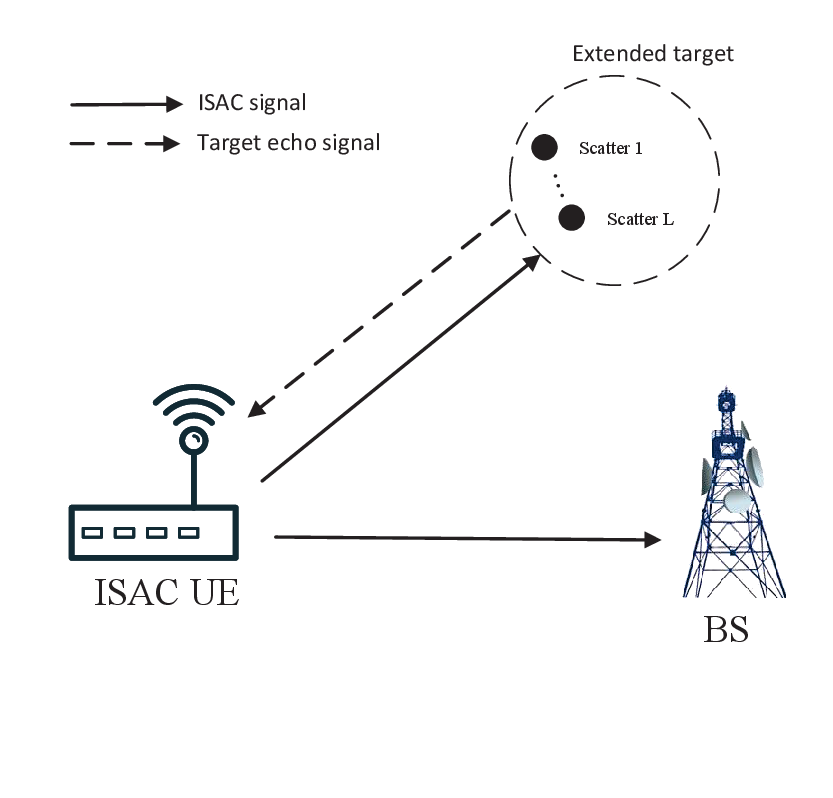}}
	\caption{Considered MIMO ISAC system.}
	\label{System}
\end{figure}
\subsection{Signal Model}
We consider an ISAC system as shown in Fig.~\ref{System}, where an ISAC UE is equipped with $N_t$ transmit  antennas and $N_r$ receive  antennas. {The UE simultaneously transmits data  to the BS equipped with $N_u$ receive antennas and senses an extended target with  $L$ scatters uniformly distributed in the proximity of the center.} {The antennas at both UE and BS are assumed to be large dimensional antenna array.}  The transmit beamforming matrix of the UE is $\mb{W}\in\mathbb{C}^{N_t \times M}$.  The  data symbol matrix is denoted as $\mb{S}\in\mathbb{C}^{M\times N_s}$, where $N_s$ denotes the number of signal samples and $M$ denotes the number of data steams satisfying $M \leq N_t$. It is assumed that the signal vectors of $\mb{S}$ are statistically orthogonal to each other, i.e., $\mbb{E}[\mb{S} \mb{S}^\dagger ]=\mb{I}_{M}$, with the notation $(\cdot)^\dag$ denoting the conjugate transpose operation. 

The received signal at BS can be expressed as
\begin{align}\label{receive_c}
\mb{Y}_c=\mb{H}_c\mb{WS} + \mb{N}_C,
\end{align}
where $\mb{H}_c$ is the channel between UE and BS,  $\mb{N}_C\in\mathbb{C}^{N_u\times N_s}$ is the additive white Gaussian noise (AWGN) at the receiving antennas of BS and $\mb{N}_C\sim\mc{CN}(0,\sigma_{c}^{2}\mb{I}_{N_s})$.
The received echoes at the UE, $\mb{Y}_s\in\mathbb{C}^{N_r}$, can be expressed as
\begin{align}\label{receive_s}
\mb{Y}_s=\sum_{l=1}^L\mb{G}_{l}\mb{WS} + \mb{N}_S,
\end{align}
where $\mb{G}_{l}\in\mathbb{C}^{N_r\times N_t}$ is the round-trip channel matrix between the $l-$th scatter and UE, $\mb{N}_S\in\mathbb{C}^{N_r\times N_s}$ is the AWGN at the receiving antennas of the UE and  $\mb{N}_S\sim\mc{CN}(0,\sigma_{s}^{2}\mb{I}_{N_s})$.
\subsection{Channel Model}
{Due to the limitations of the Kronecker channel model in accurately representing the correlation between transceivers when scattering clusters are uniformly distributed, we employ the Weichselberger MIMO channel model \cite{weichselberger2006stochastic}, which captures the spatial correlation at both ends of the link and their mutual interdependence, for both radar sensing channel and the communication channel.}
The channels in \eqref{receive_c} and \eqref{receive_s} can be expressed as
\begin{align}
\mb{H}_c & = \overline{\mb{H}}_c + \widetilde{\mb{H}}_c = \overline{\mb{H}}_c + \mb{U} (\mb{M} \odot \mb{P}) \mb{V}^{\dagger},\label{H_c}\\
\mb{G}_l &= \overline{\mb{G}}_l + \widetilde{\mb{G}}_l = \overline{\mb{G}}_l + \mb{R}_l (\mb{N}_l\odot\mb{Q}_l) \mb{T}_l^{\dagger}, 1\le l\le L,\label{Gl}
\end{align}
where $\overline{\mb{H}}_c$ and $\overline{\mb{G}}_l$ are deterministic matrices which denote the line-of-sight (LoS) component of $\mb{H}_c$ and $\mb{G}_l$. $\mb{U}$, $\mb{V}$, $\mb{R}_l$ and $\mb{T}_l$ are deterministic unitary matrices. $\mb{M}$ and $\mb{N}_l$ are deterministic nonnegative matrices which represent the variance profiles of the random components $\widetilde{\mb{H}}_c$ and $\widetilde{\mb{G}}_l$ respectively. $\mb{P}\in \mathbb{C}^{N_u\times N_t}$ and $\mb{Q}_l \in \mathbb{C}^{N_r\times N_t} $ are complex Gaussian distributed with  $[\mb{P}]_{i,j}\sim\mc{CN}(0,1/N_t)$ and $[\mb{Q}_l]_{i,j}\sim\mc{CN}(0,1/N_t)$.

\subsection{ { Problem Formulation}}
From now on, for convenience of expression, we use $\widetilde I_s({\sigma}\!^2)$ and $\widetilde I_c({\sigma}\!^2)$ to denote $\widetilde I_s({{\sigma_s}\!^2})$ and $\widetilde I_c({\sigma_c}\!^2)$, respectively. According to \cite{xu2023joint}, the sensing MI can be expressed as
\begin{align}
\widetilde{I_s}({\sigma}^2) \notag \! &= \mtm{logdet}\left(\mb{I}_{N_s}+ \frac{1}{\sigma^2}\sum_{l=1}^{L}(\mb{S}^\dagger\mb{W}^\dagger\mb{G}_l^\dagger \mb{G}_l\mb{WS})\right)\notag \\
&= \mtm{logdet}\left(\mb{I}_{LN_r}+ \frac{1}{\sigma^2}(\hat{\mb{G}}\mb{S}\mb{S}^\dagger\hat{\mb{G}}^\dagger)\right), 
\end{align}
{where  the matrix $\hat{\mb{G}}$ is defined as $\hat{\mb{G}}=[{(\bf{G}}_{1}\mb{W})^\dagger, {(\bf{G}}_{2}\mb{W})^\dagger,..., {(\bf{G}}_{l}\mb{W})^\dagger]^\dagger$.}
Based on the received signal at BS, the communication MI for the considered ISAC system can be expressed as \cite{tulino2004random}
\begin{align}
\widetilde I_c(\sigma^2)  &= \mtm{logdet}(\mb{I}_{N_u}+ \frac{1}{\sigma^2}{\mb{H}_c}\mb{Q}\mb{H}_{c}^\dagger),
\end{align} 
where $\mb{Q}=\mb{W}\mb{W}^\dagger$.
In order to achieve a balance between communication performance and sensing performance, we define a weighted MI as the performance metric for the proposed ISAC system, which is expressed as
\begin{align}\label{WMI}
\widetilde I(\mb{W}) = \rho \widetilde I_s(\sigma^2)  + (1-\rho) \widetilde I_c(\sigma^2) ,
\end{align}
where $\rho$ is a weighting factor that determines the weights of communication performance and sensing performance in the ISAC system.
Furthermore, the transmit beamforming problem can be formulated as
\begin{subequations}\label{P1}
	\begin{align}
	\left ({\textbf {P1} }\right)&\max \limits _{ {{{\bf{W}}}}}~~	\widetilde I({\bf{W}})    \label{P1a}  
	\\&  \mathrm {s.t.} \quad  {{\left\| \mb{W}  \right\|}_{F}}^{2} \leq P_t, \label{power_1}
	\end{align}	
\end{subequations}
where ${{\left\| \cdot  \right\|}_{F}}\!^{2}$ denotes the Frobenius norm, (\ref{power_1}) is the transmit power constraint and $P_t$ is the transmit power budget.

{According to \cite[Theorem 2]{Hoydis2011Asymptotic}, we have
	\begin{align}
	&\frac{1}{LN_r}\widetilde{I_s}({\sigma}^2) \xrightarrow[LN_r\rightarrow \infty]{a.s.}  {I_s}({\sigma}^2) \label{I_s_as1}\\
	&\frac{1}{N_u}\widetilde{I_c}({\sigma}^2) \xrightarrow[N_u\rightarrow \infty]{a.s.}  {I_c}({\sigma}^2)
	\end{align}
	where the notation  $\xrightarrow[LN_r\rightarrow \infty]{a.s.} $ denotes the almost surely convergence as $LN_r$ and $N_t$ tend to infinity with the ratio $LN_r/N_t$ fixed  and $\xrightarrow[N_u\rightarrow \infty]{a.s.}$ denotes the almost surely convergence as $N_u$ and $N_t$ tend to infinity with the ratio $N_u/N_t$ fixed. ${I_s}({\sigma}^2)$ and ${I_c}({\sigma}^2)$ are both deterministic,  {depending on the statistics of $\mb{G}_l$ and $\mb{H}_c$, respectively.}} However, since the  channel statistics are not perfectly known at UE, we aim to optimize the weighted asymptotic MI, which  is defined as
\begin{align}\label{asymptoticWMI}
I(\mb{W}) = \rho LN_r I_s(\sigma^2)  + (1-\rho) N_u I_c(\sigma^2) ,
\end{align}
Furthermore, the corresponding transmit beamforming problem  can be formulated as:
\begin{subequations}\label{P2}
	\begin{align}
	\left ({\textbf {P2} }\right)&\max \limits _{ {{{\bf{W}}}}}~~	I({\bf{W}})    \label{P2a}  
	\\&  \mathrm {s.t.} \quad  {{\left\| \mb{W}  \right\|}_{F}}^{2} \leq P_t. \label{power_2}
	\end{align}	
\end{subequations}
 To solve this problem, we will derive the closed-form expression for  the  weighted asymptotic MI $I(\mb{W})$ in Section \ref{close-form}.

\section{Problem Reformulation and Proposed Algorithm} \label{close-form}
{In this section, we reformulate  Problem $	\left ({\textbf {P1} }\right)$ and propose an efficient algorithm to solve the reformulated problem.  Specially, in Section~\ref{Formu1}, we first utilize the free probability theory and the  linearization trick to derive the closed-form expression of the Cauchy transform. Based on this, we reformulate  Problem $	\left ({\textbf {P1} }\right)$ by deriving the closed-form expression for the weighted MI of the considered ISAC system. Finally, the PGA algorithm is proposed to solve the reformulated problem in Section~\ref{Formu2}.}

\subsection{{Problem Reformulation} }\label{Formu1}
 Denote $\mb{B}_1=\hat{\mb{G}}\mb{S}\mb{S}^\dagger\hat{\mb{G}}^\dagger$ and  define the empirical cumulative distribution function (ECDF) of the $B_1$'s eigenvalues as
 \begin{align}
\widetilde{F}_{\mathbf{B}_1}(x)=\frac{1}{LN_r}\sum_{i=1}^{LN_r}1\{\lambda_i(\mathbf{B}_1)\leq x\},
 \end{align}
 where $\lambda_1(\mathbf{B}_1)$,...,$\lambda_{LN_r}(\mathbf{B}_1)$ are the eigenvalues of $\mb{B}_1$ and $1\{ \cdot   \}$ is the indicator function.  
Therefore, $\frac{1}{LN_r} \widetilde I_s(\sigma^2)$ can be rewritten as
\begin{align}
\frac{1}{LN_r} \widetilde{I}_s({\sigma}^2) &=  \frac{1}{LN_r} \mtm{logdet}\left(\mb{I}_{LN_r}+ \frac{1}{\sigma^2}\mb B_1\right), \notag \\
	& = \sum_{i=1}^{LN_r} \log \left(1+ \frac{1}{\sigma^2} \lambda_i\left(\mb{B}_1\right)\right),  \notag \\
& =\int_0^{\infty} \log (1+\frac{1}{\sigma^2} x) d \widetilde {F}_{\mb{B}_1}(x).
\end{align}
For $\mb B_1$,  its resolvent  $\widetilde{\mc{G}}_{\mb{B}_1}(-z)$ and  Cauchy transform ${\mc{G}}_{\mb{B}_1}(-z)$ for $\mb B_1$  are defined as
\begin{align}
\widetilde {\mc{G}}_{\mb{B}_1}(z) = \frac{1}{LN_r}\mtm{Tr}(z\mb I - \mb B_1)^{-1}=\int_{0}^{\infty} \frac{1}{z-\lambda}\mathrm{d}\widetilde F_{\mb{B}_1}, \\
{\mc{G}}_{\mb{B}_1}(z) = \frac{1}{LN_r} \mbb{E} \left[{\mtm{Tr}(z\mb I - \mb B_1)^{-1}}\right ]=\int_{0}^{\infty} \frac{1}{z-\lambda}\mathrm{d}F_{\mb{B}_1}.
\end{align}
where $ z = -\frac{1}{\sigma^2}$ and $F_{\mb{B}_1}$ is the cumulative distribution function (CDF) of the eigenvalues of $\mb{B}_1$. Therefore, we have
\begin{align}\label{cauchy1}
\frac{1}{LN_r} \frac{d \widetilde{I}_s(\sigma^2)}{dz} = -\frac{1}{z}- \widetilde {\mc{G}}_{\mb{B}_1}(-z).
\end{align}
Applying the relationship between the Shannon transform $\mc{V}_{\mb{B}_1}(z)$ and the Cauchy transform $\mc{G}_{\mb{B}_1}(-z)$, the following equation holds
\begin{align}
\frac{d\mc{V}_{\mb{B}_1}(z)}{dz}	= -\frac{1}{z}- \mc{G}_{\mb{B}_1}(-z).
\end{align}
According to \cite{Silverstein1995Strong}, the resolvent $\widetilde {\mc{G}}_{\mb{B}_1}(-z)$ converges almost surely to the Cauchy transform ${\mc{G}}_{\mb{B}_1}(-z)$ for $\mb B_1$. Therefore, we have  
\begin{align}
\frac{1}{LN_r}\widetilde{I_s}({\sigma}^2) \xrightarrow[LN_r\rightarrow \infty]{a.s.} \mc{V}_{\mb{B}_1}(z).
\end{align} 
Since \eqref{I_s_as1} holds, we can rewrite ${I_s}({\sigma}^2)$ as
\begin{align}
{I_s}({\sigma}^2) = \mc{V}_{\mb{B}_1}(z).
\end{align} 

Similarly, for the asymptotic communication MI, we have
\begin{align}
&{I_c}({\sigma}^2) = \mc{V}_{\mb{B}_2}(z), \\
&	\frac{d\mc{V}_{\mb{B}_2}(z)}{dz}	=-\frac{1}{z}-\mc{G}_{\mb{B}_2}(-z),\\
&\mc{G}_{\mb{B}_2}(z) = \int_{0}^{\infty} \frac{1}{z-\lambda}\mathrm{d}F_{\mb{B}_2}(\lambda)
\end{align}
where ${\mb{B}_2}={\mb{H}_c}\mb{Q}\mb{H}_{c}^\dagger$, $F_{\mb{B_2}}(\lambda)$ is CDF of $\mb{B}_2$,     $\mc{V}_{\mb{B}_2}(\sigma^2)$ is the Shannon transform of $\mb{B}_2$ and $\mc{G}_{\mb{B}_2}(z)$ is the Cauchy transform of $\mb{B}_2$.

To obtain the closed-form expression of $\mc{G}_{\mb{B}_i}(z)$, free probability theory serves as a powerful  analytical tool \cite{belinschi2017analytic}. However, in the considered system, $\hat{\mb{G}}$ and $\mb{S}$ are not free in the classic free probability aspect, {resulting in difficulties in directly obtaining the Cauchy Transform for the product of $\hat{\mb{G}}$, $\mb{S}$, $\mb{S}^\dagger$ and $\hat{\mb{G}}^\dagger$. To address this issue, the linearization trick will be adopted to embed the non-free matrices into a lager matrix, in which the  deterministic parts and random parts are proved to be asymptotically free. Then the desired Cauchy transform can be obtained  from the transformation of the operator-valued Cauchy transform for the embedded matrix.} For the convenience of expression, we define $\overline{\mb{G}} = [{(\overline{\bf{G}}}_{1}\mb{W})^T, {(\overline{\bf{G}}}_{2}\mb{W})^T,..., {(\overline{\bf{G}}}_{l}\mb{W})^T]^T$ and  $\overline{\mb{H}} = \overline{\mb{H}}_c\mb{W}$. Since $I_c$ and $I_s$ share a similar form, we only provide the detailed derivation of $\mc{G}_{\mb{B}_1}(z)$.

To obtain the closed-form of $\mc{G}_{\mb{B}_1}(z)$, we first apply the Anderson's linearization trick \cite{belinschi2017analytic} to construct a block matrix of size $(LN_r+N_s+2M )\times (LN_r+N_s+2M )$ as
{\begin{align}
	\setlength{\arraycolsep}{1pt}
	\def\arraystretch{1}
	\mb{B}_L  \!\!=\! \!\left[\begin{array}{cccc}
	\mb{0}_{LN_r\times LN_r} & \mb{0}_{LN_r\times M} & \mb{0}_{LN_r\times N_s} & \hat{\mb{G}}\\
	\mb{0}_{{M}\times LN_r} & \mb{0}_{{M}\times {M}} & \mb{S} & -\mb{I}_M\\
	\mb{0}_{N_s\times LN_r} & \mb{S}^\dagger & -\mb{I}_{N_s} & \mb{0}_{N_s\times M}\\
	\hat{\mb{G}}^{\dagger} & -\mb{I}_M & \mb{0}_{M\times N_s} & \mb{0}_{{M}\times {M}}
	\end{array}\right].\label{eqBL}
	\end{align}	
	The operator-valued Cauchy transform $\mc{G}_{\mb{B}_L}^{\mc{D}}$ of $\mb{B}_L$  \cite{belinschi2017analytic} is given by
	\begin{align}
	\mc{G}_{\mb{B}_L}^{\mc{D}}(\boldsymbol{\Lambda}(z)) = \mbb{E}_{\mc{D}}\left[\left(\boldsymbol{\Lambda}(z) - \mb{B}_L\right)^{-1}\right],\label{eqGBL0}
	\end{align}
	where $\mbb{E}_\mc{D}\left[\mb{X}\right]$ is defined as 
	\begin{align}
	\setlength{\arraycolsep}{1.5pt}
	\def\arraystretch{1}
	\mbb{E}_\mc{D}\left[ \mb{X}\right] \!=\! \left[\begin{array}{c|c|c|c}
	\mbb{E}\left[\mb{X}_1\right] & & & \\\hdashline
	& \mbb{E}\left[\mb{X}_2\right] & & \\\hdashline
	& & \mbb{E}\left[\mb{X}_3\right] & \\\hdashline
	& & &  \mbb{E}\left[\mb{X}_4\right]
	\end{array}\right],\label{eqEDX}
	\end{align}
	where $\mb{X}_1=\{\mb{X}\}_1^{LN_r}$, $\mb{X}_2=\{\mb{X}\}_{LN_r+1}^{LN_r+M}$, $\mb{X}_3=\{\mb{X}\}_{LN_r+M+1}^{LN_r+M+N_s}$, $\mb{X}_4=\{\mb{X}\}_{LN_r+M+N_s+1}^{LN_r+2M+N_s}$ ,with the notation 
	$ { \{ \mb{A} \} ^{b}_{a}}    $ denoting the submatrix of $\bf A$ containing the rows and columns with indices from $a $ to $ b $, i.e.,  $  \left[\{ \mb{A} \} ^{b}_{a}\right]_{i,j} = \left[ \mb{A}\right]_{i+a-1,j+a-1}  $ for $ 1\le i,j\le b-a+1$, where the notation $ [\mb{A}]_{i,j} $ is the element in the $i$-th row and $j$-th column of matrix $ \mb{A} $. 	
	The matrix function $ \boldsymbol{\Lambda}(z) $ is defined as
	\begin{align}
	\boldsymbol{\Lambda}(z) = \begin{bmatrix}
	z\mb{I}_{LN_r} & \mb{0}_{LN_r\times (N_s+2M)}\\
	\mb{0}_{(N_s+2M)\times LN_r} & \mb{0}_{(N_s+2M)\times (N_s+2M)} 
	\end{bmatrix}.\label{eqLambdaz}
	\end{align}
	Then $\mc{G}_{\mb{B}_1}(z)$ is given by
	\begin{align}
	\mc{G}_{\mb{B}_1}(z) = \frac{1}{LN_r}\mathrm{Tr}\left(\left\{\mc{G}_{	\mb{B}_L}^{\mc{D}}(\boldsymbol{\Lambda}(z))\right\}^{(1,1)}\right),\label{eqGB_GL}
	\end{align}
	where $\{\cdot\}^{(1,1)}$ represents the upper-left $LN_r\times LN_r$ matrix block and the operator $\text{Tr}(\cdot)$ represents the trace of the matrix. It can be observed that $\mb{B}_L$ shares a similar structure with the matrix $\mb{L}$ proposed in \cite[Prop. 2]{zheng2023mutual}, which inspires us to use  the method in \cite{zheng2023mutual}. Specifically, the  Cauchy 
	transform $\mc{G}_{\mb{B}_1}(z)$ can be obtained with the following proposition.
	\begin{proposition}\label{prop_cauchyB1}
		The Cauchy transform of $\mb{B}_1$, with $z\in\mbb{C}^+$, is given by
		\begin{align}\label{eqGB}
		\mc{G}_{\mb{B}_1}(z) = \frac{1}{LN_r}\mathrm{Tr}\left[	\mc{G}_{\widetilde{\mb{C}}}(z)\right],
		\end{align}
		where 	$ \mc{G}_{\widetilde{\mb{C}}}(z) $ satisfies the following equations
		{\begin{align}
			\mc{G}_{\widetilde{\mb{C}}}(z) &= \left(\widetilde{\boldsymbol{\Psi}}(z) - \overline{\mb{G}} {\Pi}^{-1}\overline{\mb{G}}^\dagger \right)^{-1},\label{eqGCt}\\
			\boldsymbol{{\Pi}} & =\boldsymbol{\Psi}(z)-\widetilde{\boldsymbol{\Phi}}(z)^{-1},\label{eqpi}
			\end{align}}
		where the matrices  $	\widetilde{\boldsymbol{\Psi}}(z)$, $	\boldsymbol{\Psi}(z)$, $	\widetilde{\boldsymbol{\Phi}}(z)$,
		$\boldsymbol{\Phi}(z)$ are respectively denoted as
		\begin{align}
		\widetilde{\boldsymbol{\Psi}}(z) &= z\mb{I}_{LN_r} -\mathrm{diag}\left\{   \widetilde{\eta}_1(\mc{G}_{\mb{C}}(z)),\  \ldots,\  \widetilde{\eta}_L(\mc{G}_{\mb{C}}(z)\right\} ,\label{eqPsis}\\
		\boldsymbol{\Psi}(z) &=- \sum_{l=1}^L \eta_l(\mc{G}_{{\widetilde{\mb{C}}}_l}(z)),\label{eqPsi1}\\
		\widetilde{\boldsymbol{\Phi}}(z) &= -\widetilde{\zeta}(\mc{G}_{\widetilde{\mb{D}}}(z)),\\
		\boldsymbol{\Phi}(z) &= \mb{I}_{N_s}-\zeta(\mc{G}_{\mb{D}}(z)),\label{eqPhi}
		\end{align}	
		where the notation  $\text{diag}(\mb{A},\cdots,\mb{B})$ represents the diagonal block matrix constructed by $  \mb{A},\cdots,\mb{B} $ matrices and $\eta_l(\widetilde{\mb{C}})$, $\widetilde{\eta}_l({\mb{C}})$, $\zeta(\mb{D})$, $\widetilde{\zeta}(\widetilde{\mb{D}})$ are the parameterized one-sided correlation matrices, which are shown as \eqref{A1}, \eqref{A2}, \eqref{A9} and \eqref{A10} in Appendix \ref{PARE}. The matrices 
		$\mc{G}_{\mb{C}}(z)$, $\mc{G}_{\widetilde{\mb{D}}}(z)$, $\mc{G}_{\mb{D}}(z)$, and $\mc{G}_{{\widetilde{\mb{C}}}_l}(z)$ are defined as
		
		\begin{align}
		\mc{G}_{\mb{C}}(z) &\!= \left(\boldsymbol{\Psi}(z) - \overline{\mb{G}}^{\dagger} \widetilde{\boldsymbol{\Psi}}(z)^{-1}\overline{\mb{G}}^{\dagger} - \widetilde{{\boldsymbol\Phi}}(z)^{-1} \right)^{-1}, \label{eqGC1} 
		\\
		\mc{G}_{\widetilde{\mb{D}}}(z) &= \boldsymbol{\Phi}(z)^{-1},\label{eqGDt} 
		\\
		\mc{G}_{\mb{D}}(z) &= \left(\widetilde{\boldsymbol{\Phi}}(z) - \left(\boldsymbol{\Psi}(z) -  \overline{\mb{G}}^{\dagger}\widetilde{\boldsymbol{\Psi}}(z)\overline{\mb{G}} \right)^{-1} \right)^{-1} ,\label{eqGDk}\\
		\mc{G}_{{\widetilde{\mb{C}}}_l}(z) &= \{\mc{G}_{\widetilde{\mb{C}}}(z)\}_{1+(l-1)N_r}^{lN_r}.\label{eqGC2} 
		\end{align}
		
	\end{proposition}
	\begin{proof}
		The proof of this Proposition 1 is similar to the method presented in \cite[Prop. 2]{zheng2023mutual}. Therefore we 		provide  a brief outline of the proof. First, we prove that the deterministic and random components of the linearized matrix are free. Then, by applying the subordination formula, we derive the equation for the operator-valued Cauchy transform. Finally, using the matrix inversion formula, we decompose the operator-valued Cauchy transform, thereby obtaining the above expressions. 
	\end{proof}	
	Based on the closed-form expression of  the Cauchy transformation  $\mc{G}_{\mb{B}_1}(z)$ , we can derive the closed-form expression of the Shannon transform $ \mc{V}_{\mb{B}_1}(z) $ in the following proposition.

	\begin{proposition}\label{prop_shano1}
		The Shannon transform $ \mc{V}_{\mb{B}_1}(z) $, with $z\in\mbb{C}^+$, is given by 
		\begin{align}
			\mc{V}_{\mb{B}_1}(z)&=\frac{1}{LN_r}\log  {\det \left( {\frac{{{\boldsymbol{\widetilde \Psi }}\left(-z\right)}}{{ - z}}} \right)} + \frac{1}{LN_r}\log{\det \left( \boldsymbol{\Psi}(-z) - \overline{\mb{G}}^\dagger {\widetilde{\boldsymbol \Psi}}(-z)^{-1} \overline{\mb{G}}^\dagger - {\widetilde{\boldsymbol{\Phi}}}(-z)^{-1}  \right)}  \notag \\		
		&+\frac{1}{LN_r}\log \det\left( {\widetilde{\boldsymbol{\Phi}}}(-z)
		\right) + \frac{1}{LN_r}\mtm{Tr}\left( \mc{G}_{\widetilde{\mb{C}}}(-z) \left(-z\mb{I}_{LNr}- \widetilde{\boldsymbol{\Psi}}(-z) \right) \right),  \notag\\
		&+\frac{1}{LN_r}\mtm{Tr}\left( \mc{G}_{\widetilde{\mb{D}}}(-z) {\zeta}  \right) + \frac{1}{LN_r}\log \det \left( \boldsymbol{\Phi}(-z) \right) \label{VB1}
		\end{align}
		
		where  $	\widetilde{\boldsymbol{\Psi}}(-z)$, $	\boldsymbol{\Psi}(-z)$, $	\widetilde{\boldsymbol{\Phi}}(-z)$,
		$\mc{G}_{\widetilde{\mb{D}}}(-z)$ , $	\mc{G}_{\mb{D}}(-z) $, and $\mc{G}_{\widetilde{\mb{C}}}\left(-z\right)$ are given by
		\eqref{eqGCt}-\eqref{eqPhi} in Proposition~\ref{prop_cauchyB1}.
	\end{proposition}
	\begin{proof}
		The proof of Propsition~\ref{prop_shano1} is given in Appendix \ref{appx_shano}  .
	\end{proof}

	Similarly, the Cauchy transform $\mc{G}_{\mb{B}_2}(z)$ of $\mb{B}_2$ and the Shannon transform $ \mc{V}_{\mb{B}_2}(z) $, with $z\in\mbb{C}^+$, are given by 
	\begin{align}\label{eqGB2}
	\mc{G}_{\mb{B}_2}(z) &= \frac{1}{N_u}\mathrm{Tr}\left[	\mc{G}_{\widetilde{\mb{E}}}(z)\right],\\
	\mc{V}_{\mb{B}_2}(z)&=\frac{1}{N_u}\log  {\det \left( {\frac{{{\boldsymbol{\widetilde \Omega }}\left(-z\right)}}{{ - z}}} \right)} + \frac{1}{N_u}\mtm{Tr}\left( \mc{G}_{\widetilde{\mb{E}}}(-z) \widetilde{{\tau}}  \right) \notag \\
	&-\frac{1}{N_u}\log \det\left(\mc{G}_{\mb{E}}(-z)\right),	 \label{VB2}
	\end{align}
	where $\eta_l(\widetilde{\mb{C}})$, $\widetilde{\eta}_l({\mb{C}})$, $\tau(\mb{E})$, $\widetilde{\tau}(\widetilde{\mb{E}})$ are the parameterized one-sided correlation matrices, which are shown as \eqref{A1}, \eqref{A2}, \eqref{A5} and \eqref{A6} in Appendix \ref{PARE}. The matrices  	$ \mc{G}_{\widetilde{\mb{E}}}(z) $, $\mc{G}_{\mb{E}}(z)$, $\widetilde{\boldsymbol{\Omega}}(z)$ and $ {\boldsymbol{\Omega}}(z)$  satisfy the following equations
	{\begin{align}
		\mc{G}_{\widetilde{\mb{E}}}(z) &= \left(\widetilde{\boldsymbol{\Omega}}(z) - \overline{\mb{H}} {\boldsymbol{\Omega}}(z)^{-1}\overline{\mb{H}}^\dagger \right)^{-1},\label{eqGEt}\\
		\mc{G}_{\mb{E}}(z) &= \left({\boldsymbol{\Omega}}(z) - \overline{\mb{H}}^\dagger \widetilde{\boldsymbol{\Omega}}(z)^{-1}\overline{\mb{H}} \right)^{-1},\label{eqGC}\\
		\widetilde{\boldsymbol{\Omega}}(z) &= z\mb{I}_{N_u} - \widetilde{\tau}(\mc{G}_{\mb{E}}(z)) ,\label{eqOmet}\\
		\boldsymbol{\Omega}(z) &= 1 - \tau(\mc{G}_{\widetilde{\mb{E}}}(z)).\label{eqOme}
		\end{align}}
{	From now on, for convenience of expression, we refer to the asymptotic sensing MI, asymptotic communication MI and weighted asymptotic MI as the sensing MI,  communication MI and  weighted MI, respectively.}	Then the sensing MI , the communication MI and the weighted MI can be rewritten as 
	\begin{align}
	I_s(z,{\bf{W}})&=LN_{r} \mc{V}_{\mb{B}_1}(\sigma^2),  \\
	I_c(z,{\bf{W}})&=N_u \mc{V}_{\mb{B}_2}(\sigma^2), \\
	I(z,{\bf{W}}) &= \rho   I_s(z,{\bf{W}}) + (1-\rho)  I_c(z,{\bf{W}})\label{IZS}.
	\end{align}
	Consequently, we reformulate the problem $	\left ({\textbf {P2} }\right)$ as
	\begin{subequations}
		\begin{align}
		\left ({\textbf {P3} }\right)&\max \limits _{ \bf{W}    }~~	I(z,{\bf{W}}) \label{obj}
		\\&  \mathrm {s.t.} \quad  \eqref{power_1}.
		\end{align}	
	\end{subequations}

	\subsection{Proposed PGA Algorithm} \label{Formu2}
	In order to solve the problem $\left ({\textbf {P3} }\right)$  , we propose the PGA algorithm. Firstly, we  calculate the gradient of the weighted MI in \eqref{IZS}. Then the updated beamforming  matrix at the $(i+1)$-th iteration is  updated as
	\begin{align}
	{\mb{\widetilde{W}}}^{\left[ {i+1} \right]}={{\mb{W}}^{\left[ i \right]}}+\lambda \nabla {{g}_\mb{W}}\left( {{\mb{W}}^{\left[ i \right]}} \right),
	\end{align}
	where ${\mb{W}}^{\left[ i \right]}$ denotes the beamforming matrix at the $i$-th iteration,  $\lambda$ denotes the step size, and the element of gradient  $\nabla {{g}_\mb{W}}$ of $I(z,{\bf{W}})$ is given by \eqref{GW},
	{where $\mb{E}_{i,j}$ is a matrix whose elements satisfy
		{\begin{align}
			[\mb{E}_{i,j}]_{s,t}=\left\{\begin{matrix}
			1, \qquad     \text{if} \; s=i \;  \text{and} \;  t=j,  \\
			0, \qquad   \qquad \; \; \;    \text{otherwise}.  \\
			\end{matrix} \right. 
			\end{align}}}
	To satisfy the transmit power constraint \eqref{power_1}, the solution $	{\mb{\widetilde{W}}}^{\left[ {i+1} \right]}$ is then projected onto the feasible region. The obtained solution at the $(i+1)$-th iteration is given by
	{\begin{align} \label{PROJ}
		{\mb{{W}}}^{\left[ {i+1} \right]} = \text{Proj}_{W}\left(\mb{\widetilde{W}}^{\left[ {i+1} \right]}\right),
		\end{align}}
	where the projection operator is given as
	{\begin{align}
		\text{Proj}_{W}=\left\{ \begin{matrix}
		\mb{W}, \qquad  \qquad \qquad if{{\left\| \mb{W}  \right\|}_{F}}^{2}\le P_t,  \\
		\sqrt{P_t}\frac{\mb{W}}{{{\left\| \mb{W} \right\|}_{F}}}, \qquad  \qquad  otherwise.  \\
		\end{matrix} \right.
		\end{align}}
	The detailed algorithm is presented in Algorithm 1.
	\begin{algorithm}\label{op_theta}
		\caption{PGA algorithm for beamforming matrix design} 
		\begin{algorithmic}
			\State \textbf{initialize}:	Set the convergence criterion $ \varepsilon $, $i=0$, 
			and randomly generate the beamforming matrix 	$ \mb{W}^{[0]}$. Calculate $ I^{[0]}$ based on \eqref{IZS}. 
			\State	\textbf{repeat} 
			\State \quad Update ${\mb{{W}}}^{\left[ {i+1} \right]}$ based on \eqref{PROJ}.
			\State \quad Calculate $ I^{[i+1]}$ based on \eqref{IZS}.
			\State \quad $ i=i+1 $.
			\State  \textbf{until} 	$ \left|I^{[i]} - I^{[i-1]} \right| \le \varepsilon $.
			\State  \textbf{output}  Optimal beamforming matrix $\mb{W}$.
		\end{algorithmic}
	\end{algorithm}
	
	\begin{figure*}[t]
		\centering
		
		\begin{align}
	\\
		[\nabla {{g}_\mb{W}}]_{i,j}=&\rho\mtm{Tr}\left( {{\left( -\boldsymbol \Psi +{\mb{W}^{\dagger}}{{{\overline{\mb{G}}}}}{{{\widetilde{\boldsymbol{\Psi} }}}^{-1}}\overline{\mb{G}}\mb{W}+{{{\widetilde{\boldsymbol\Phi }}}^{-1}} \right)}^{-1}} {\mb{E}_{i,j}^{\dagger}}{\overline{\mb{G}}^{\dagger}}{{{\widetilde{\boldsymbol{\Psi} }}}^{-1}}\overline{\mb{G}}\mb{W} \right) \notag\\ &+(\rho - 1) \mtm{Tr}\left({{\left(\boldsymbol{\Omega} -{\mb{W}^{\dagger}}{{{\overline{\mb{H}}}}^{\dagger}}{{{\widetilde{\boldsymbol{\Omega} }}}^{-1}}\overline{\mb{H}}\mb{W} \right)}^{-1}}{\mb{E}_{i,j}^{\dagger}}{{\overline{\mb{H}}}^{\dagger}}{{\widetilde{\boldsymbol{\Omega} }}^{-1}}\overline{\mb{H}}\mb{W}\right), \label{GW}
		\end{align}

	\end{figure*}
	
	\section{Numerical Results}\label{secResult}
	In this section, we provide numerical results to verify the accuracy of the derived closed-form weighted MI expression and the  effectiveness of the proposed PGA algorithm. For the channel settings, the deterministic components are modeled as the direct links between uniform planar arrays (UPA) equipped at both the transmitter and the receiver. In addition,  statistical characteristics parameters of the channel including the deterministic unitary matrices  $\mb{U}$, $\mb{V}$, $ \mb{R}_l$ and $\mb{T}_l $, as well as the variance matrices $\mb{M}$ and $\mb{N}_l $, are generated randomly but fixed in each Monte Carlo simulation. Unless otherwise stated, the number of scatters, data streams and signal samples are set as $L=2$, $M=N_u$, and $N_s=M$, respectively. In addition, the transmit power budget is set as $P_t=N_t$. {It should be noted that since the radar detection link distance is typically longer than the communication link distance to UE, without loss of generality, we assume that the signal-noise-ratio (SNR) at BS is 20 dB higher than the SNR at UE. Besides, the SNR mentioned in this section refers to the SNR at BS.} Furthermore, the beamforming matrix is set as $\mb{W}=\sqrt{P_t/M}\mb{I}_{M}$ without optimization. All simulation results are  generated by averaging over $10^4$ channel realizations. 
	
	We first verify the accuracy of the obtained closed-form MI expression for communication MI and sensing MI  without optimization in Fig.~\ref{mainfig1}. The number of antennas are  $N_t=N_r=N_u=16$. The  solid lines   represent theoretical results calculated by  the communication MI expression and sensing MI expression, where the  markers represent corresponding Monte Carlo simulation results. It can be observed that  the  theoretical analyses match the simulation results excellently, verifying the accuracy of our derived results. 
	
	\begin{figure}[t]
		\centering

			{\centering
				\includegraphics[width=0.6\columnwidth]{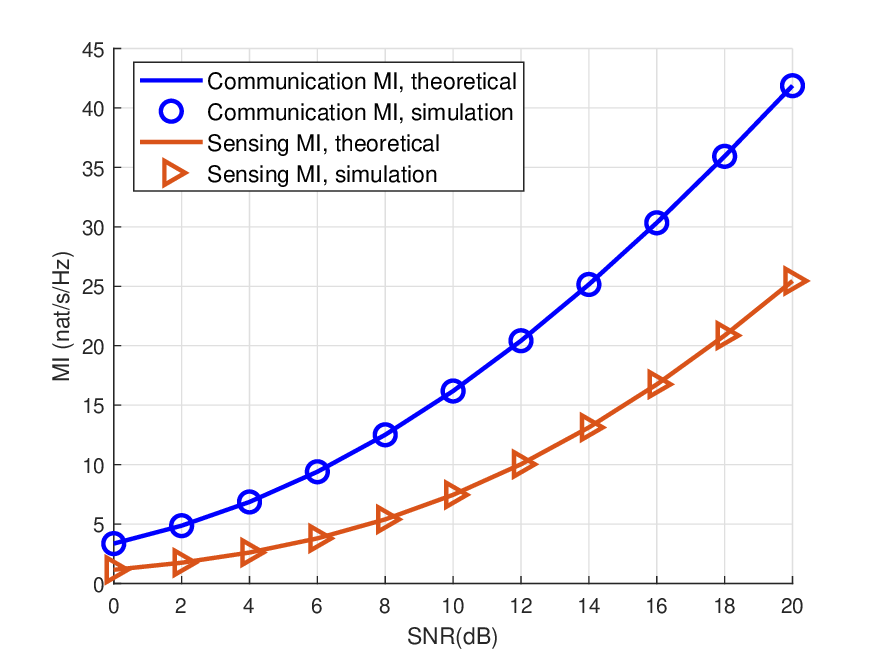}}
			\caption{Communication MI and sensing MI versus SNR.}
	\end{figure}

	Subsequently, the convergence of the proposed PGA algorithm under different numbers of antennas is plotted in Fig.~\ref{mainfig1}. The weighting factor and the SNR are set as $\rho=0.8$ and 10 dB respectively.  It can be observed that the weighted MI  increases with the number of antennas. {This is because a larger number of antennas provides higher diversity gain and greater design degrees
		of freedom (DoFs).} In addition, the weighted MI increases with the iterations and converges within 3 iterations, demonstrating the fast convergence of the proposed PGA algorithm.
	\begin{figure}[t]
		\centering
		{
			\includegraphics[width=0.6\columnwidth]{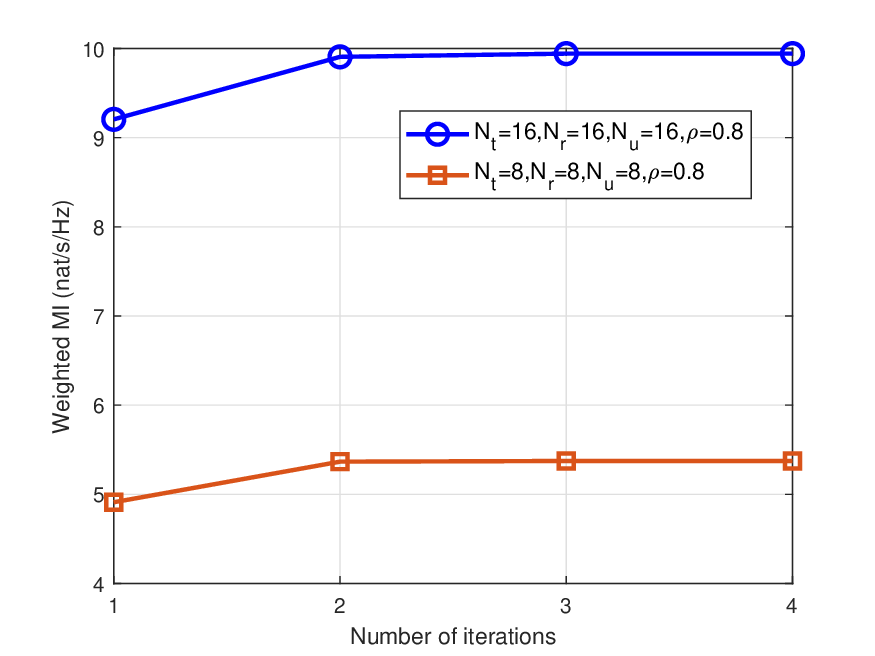}}
		\caption{Weighted MI versus the number of iterations, where $\rho = 0.8$.}
		\label{mainfig1}
	\end{figure}
	
		\begin{figure}[t]
	\centering
	{
		\includegraphics[width=0.6\columnwidth]{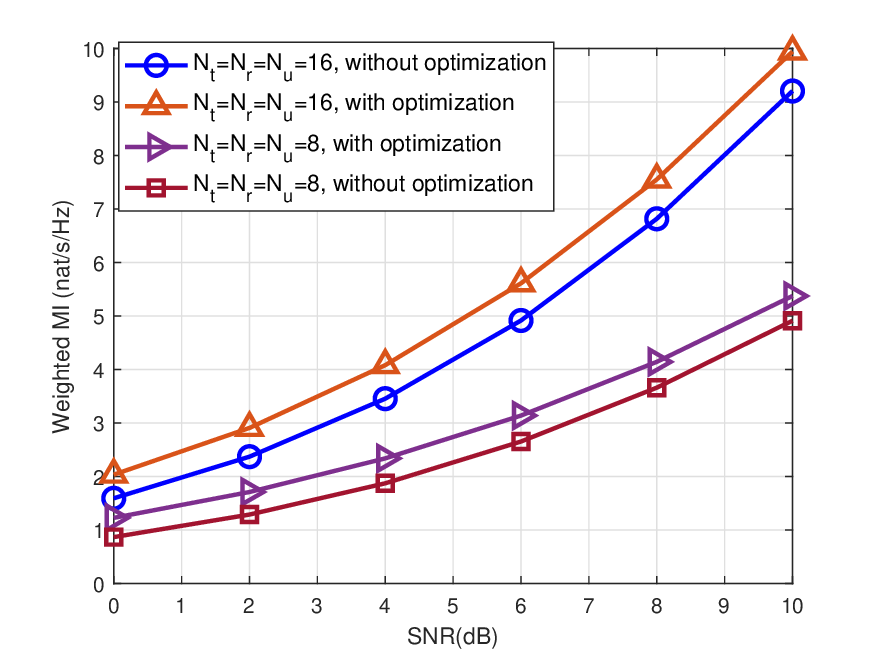}}
	\caption{ Weighted MI versus SNR with different schemes, where $\rho= 0.8$.}
\end{figure}

\begin{figure}[t]
	\centering{
		\includegraphics[width=0.6\columnwidth]{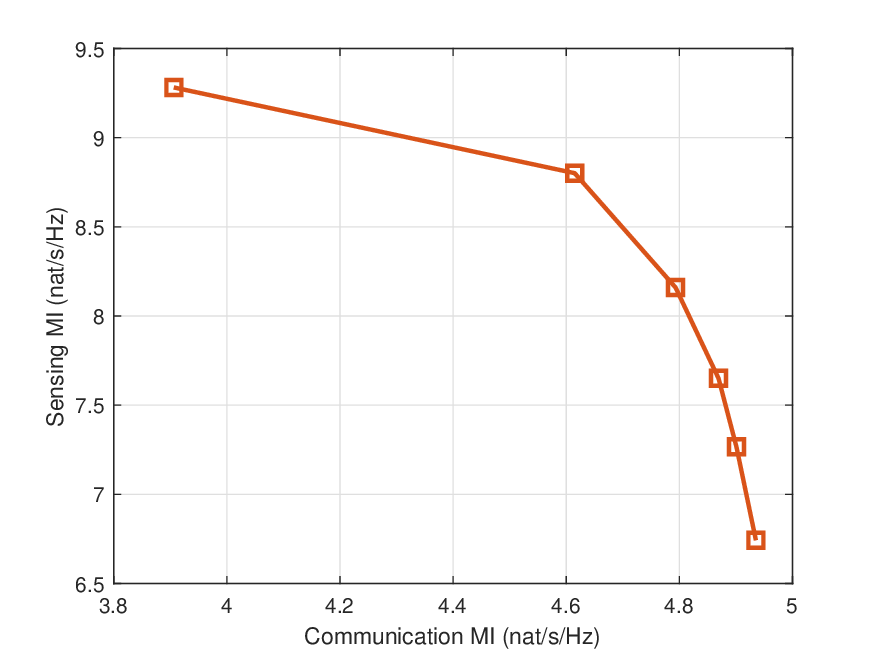}}
	\caption{Sensing MI versus Communication MI.}
	\label{mainfig2}
\end{figure}

	To demonstrate the effectiveness of the proposed optimization algorithm, we  show the optimized weighted MI with different numbers of antennas for the considered system in Fig.~\ref{mainfig2}. The weighting factor is set as $\rho=0.8$. It can be observed that with the same SNR, the proposed PGA algorithm can improve the weighted MI.


To further depict the trade-off between sensing and communication MI, we show the optimized weighted MI with weighting factors $\rho$ ranging from 0 to 1 in Fig.~\ref{mainfig2}. The number of antennas and the SNR are set as $N_t=N_r=N_u=8$ and 10 dB, respectively.  As the weighting factor increases, the ISAC system places greater emphasis on sensing performance, and conversely, prioritizes communication performance when the weighting factor decreases.

\section{Conclusions}\label{secConclude}
In this paper, we investigated the beamforming design for a practical MIMO ISAC system. To derive the closed-form expression of the weighted MI for the considered system, we resorted to the operator-valued free probability. Based on the closed-form expression, we proposed the PGA algorithm to design the transmit beamforming matrix for maximizing the weighted MI. Simulation results verified the accuracy of the derived closed-form expression and the effectiveness of the proposed algorithm.

\appendix
	\setcounter{section}{0}
	\section{Parameterized One-sided Correlation Matrices}\label{PARE}
	We  define the parameterized one-sided correlation matrices of $\mb{G}_l$ as
	\begin{align}
	\eta_l(\widetilde{\mb{C}}) &= \mbb{E}[\widetilde{\mb{G}}_l^\dagger\widetilde{\mb{C}} \widetilde{\mb{G}}_l] = \frac{1}{N_t} {\mb{T}}_l\boldsymbol{\Pi}_l(\widetilde{\mb{C}}){\mb{T}}_l^\dagger, 
	1\le l\le L, \label{A1}
	\\
	\widetilde{\eta}_l({\mb{C}}) &= \mbb{E}[\widetilde{\mb{G}}_l{\mb{C}} \widetilde{\mb{G}}_l^\dagger] = \frac{1}{N_t}\mb{R}_l \widetilde{\boldsymbol{\Pi}}_l({\mb{C}}) \mb{R}_l^\dagger, 1\le l\le L,\label{A2}
	\end{align}
	where $\widetilde{\boldsymbol{C}} \in \mbb{C}^{N_t \times N_t}$ and $\mb{C}_l\in \mbb{C}^{N_r \times N_r}$ are any Hermitian matrices, $\boldsymbol{\Pi}_l(\widetilde{\mb{C}})$ and $\widetilde{\boldsymbol{\Pi}}_l({\mb{C}})$ are diagonal matrices with the entries given by 
	\begin{align}
	\left[\boldsymbol{\Pi}_l(\widetilde{\mb{C}})\right]_{i,i} \!&= \sum_{j = 1}^{N_r} \left([\mb{N}_l]_{j,i}\right)^2 \left[\mb{R}_l^\dagger \widetilde{\mb{C}} \mb{R}_l\right]_{j,j}, 1\le i\le N_t,\label{A3}
	\\
	\left[\widetilde{\boldsymbol{\Pi}}_l({\mb{C}})\right]_{i,i} &= \sum_{j=1}^{N_t} \left([\mb{N}_l]_{i,j}\right)^2 \left[ {\mb{T}}_l^\dagger {\mb{C}} {\mb{T}}_l \right]_{j,j},  1\le i\le N_r.\label{A4}
	\end{align}
	Similarly,  the parameterized one-sided correlation matrices of $\mb{H}_c$, are defined as
	\begin{align}
	\tau(\mb{E}) &= \mbb{E}[{\widetilde{\mb{H}}_c}^\dagger\mb{E}\widetilde{\mb{H}}_c] = \frac{1}{N_t} \mb{V}\boldsymbol{\Sigma}(\mb{E})\mb{V}^\dagger, \label{A5}
	\\
	\widetilde{\tau}(\widetilde{\mb{E}}) &= \mbb{E}[\widetilde{\mb{H}}_c\widetilde{\mb{E}}{\widetilde{\mb{H}}_c}^\dagger] = \frac{1}{N_t}\mb{U} \widetilde{\boldsymbol{\Sigma}}(\widetilde{\mb{E}}) \mb{U}^\dagger, \label{A6}
	\end{align}
	where $\widetilde{\mb{E}} \in \mbb{E}^{N_t \times N_t}$ and $\mb{E}\in \mbb{C}^{N_u \times N_u}$ are arbitrary Hermitian matrices. $\boldsymbol{\Sigma}(\mb{E})$ and $\widetilde{\boldsymbol{\Sigma}}(\widetilde{\mb{E}})$ are diagonal matrices with the entries given by 
	\begin{align}
	\left[\boldsymbol{\Sigma}(\mb{E})\right]_{i,i} &= \sum_{j = 1}^{N_u} \left([\mb{M}]_{j,i}\right)^2 \left[\mb{U}^\dagger \mb{D} \mb{U}\right]_{j,j},\label{A7}
	\\
	\left[\widetilde{\boldsymbol{\Sigma}}{(\widetilde {\mb{E}}})\right]_{i,i} &= \sum_{j=1}^{N_t} \left([\mb{M}]_{i,j}\right)^2 \left[ \mb{V}^\dagger \widetilde{\mb{E}} \mb{V} \right]_{j,j}.\label{A8} \end{align}
	The parameterized one-sided correlation matrices of $\mb{S}$ are
	\begin{align}
	\zeta(\mb{D}) &= \mbb{E}[{\mb{S}}^\dagger{\mb{D}}{\mb{S}}] = \frac{1}{N_s} \mtm{Tr}(\mb{D})\mb{I}_{N_s},\label{A9}
	\\
	\widetilde{\zeta}(\widetilde{\mb{D}}) &= \mbb{E}[{\mb{S}}\widetilde{\mb{D}}{\mb{S}}^\dagger] =  \frac{1}{N_s} \mtm{Tr}(\widetilde{\mb{D}})\mb{I}_{M}, \label{A10}
	\end{align}
	where $\widetilde{\mb{D}} \! \in\! \mbb{C}^{N_s\! \times\! N_s}\! $  and ${\mb{D}}\!\in\! \mbb{C}^{M \!\times\! M}\! $ are any Hermitian matrices.	
	\section{Proof of Proposition~\ref{prop_cauchyB1}}\label{appx_shano}
	To prove the Proposition~ \ref{prop_shano1}, we need to prove that \eqref{cauchy1} holds with the $ \mc{V}_{\mb{B}_1}(z) $ given in  \eqref{VB1}. The partial derivative of $ \mc{V}_\mb{B}(z) $ with respect to $ z$  is given by
	{\begin{align}
		\frac{d}{dz} \mc{V}_{\mb{B}_1}(z)&=\frac{1}{LN_r}\frac{d}{dz}\log  {\det \left( {\frac{{{{\widetilde {\boldsymbol \Psi} }}}}{{ - z}}} \right)} + \frac{1}{LN_r}\frac{d}{dz}\log{\det \left( {\boldsymbol{\Psi}} - \overline{\mb{G}}^\dagger {\widetilde{\boldsymbol{\Psi}}}^{-1} \overline{\mb{G}}^\dagger - {\widetilde{\boldsymbol{\Phi}}}^{-1}  \right)} ,\notag \\
		& +\frac{1}{LN_r}\frac{d}{dz}\log \det\left( {\widetilde{\boldsymbol{\Phi}}}	\right) + \frac{1}{LN_r}\frac{d}{dz}\mtm{Tr}\left( \mc{G}_{\widetilde{\mb{C}}} \left(-z\mb{I}_{LNr}- \widetilde{\boldsymbol{\Psi}} \right) \right), \notag \\
		&+\frac{1}{LN_r}\frac{d}{dz}\mtm{Tr}\left( \mc{G}_{\widetilde{\mb{D}}} {\zeta}  \right) +\frac{1}{LN_r}\frac{d}{dz} \log \det \left( \boldsymbol{\Phi} \right) \label{VB}
		\end{align}}
	where notation $ (-z) $ are omitted for convenience. For a matrix function $\mb{K}(z)$, the following equations hold:
	\begin{align}
	\frac{d}{dz} \log \det \mb{K}(z)&= \mathrm{Tr} \left(\mb{K}(z)^{-1}\frac{d\mb{K}(z)}{dz} \right), \label{der1}
	\\	\mathrm{Tr}\left(\frac{d\mb{K}(z)^{-1}}{dz}    \right)	&= -\mathrm{Tr} \left(\mb{K}(z)^{-1}\frac{d\mb{K}(z)}{dz}\mb{K}(z)^{-1} \right). \label{der2}
	\end{align}
	According to \cite[Lemma1]{Wang2024}, the following equations hold:
	\begin{align}
	&	\mathrm{Tr}(\mb{A}_1\eta_l(\mb{A}_2))=\mathrm{Tr}(\mb{A}_2\widetilde{\eta}_l(\mb{A}_1)), 1\le l \le L,\label{der3}\\
	&\mathrm{Tr}(\mb{B}_1\zeta(\mb{B}_2))=\mathrm{Tr}(\mb{B}_2\widetilde{\zeta}(\mb{B}_1)). \label{der4}
	\end{align}
	Then we will simplify the  terms on the right-hand side of equation \eqref{VB}  based on \eqref{der1}, \eqref{der2}, \eqref{der3} and \eqref{der4}.
	
	For the first term of \eqref{VB}, we can obtain
	\begin{align}
	\frac{d}{dz}\log {\det \left( {\frac{{{{\widetilde {\boldsymbol \Psi} }}}}{{ - z}}} \right)} &=  
	\mathrm{Tr}\left(  \left( \frac{{\widetilde {\boldsymbol \Psi} }}{-z}\right)^{-1}  \frac{d\left( \frac{{\widetilde {\boldsymbol \Psi} }}{-z}\right)}{dz} \right) \notag \\
	&=\mathrm{Tr}\left( -\frac{1}{z}\mb{I}  + {{\widetilde {\boldsymbol \Psi} }}^{-1}\frac{d{\widetilde {\boldsymbol \Psi} }}{dz} \right) \notag \\
	&=-\frac{LN_r}{z}+\mathrm{Tr}\left( {{\widetilde  {\boldsymbol \Psi} }}^{-1}\frac{d{\widetilde {\boldsymbol \Psi} }}{dz} \right).\label{term1}
	\end{align}
	
	For the second term , the third term and the last term  in \eqref{VB}, we can get 
	\begin{align}
	\frac{d}{dz}\log \det\left(  - {{{{\widetilde {\boldsymbol{\Phi} }}}}^{ - 1}} + {\boldsymbol{\Delta}} \right)   &=\mathrm{Tr}\left(\left(  - {{{{\widetilde {\boldsymbol{\Phi}} }}}^{ - 1}} + {\boldsymbol{\Delta}} \right)^{-1}  \frac{d\left(  - {{{{\widetilde {\boldsymbol{\Phi}} }}}^{ - 1}} + {\boldsymbol{\Delta}} \right)}{dz}  \right),\label{term2_1}\\
	\frac{d}{dz}\log\det\left( {{{\widetilde {\boldsymbol{\Phi}} }}} \right) &=\mathrm{Tr}\left(  {{{\widetilde {\boldsymbol{\Phi}} }}}^{-1}\frac{d {{{\widetilde {\boldsymbol{\Phi}} }}} }{dz}\right),\label{term2}\\
	\frac{d}{dz}\log\det\left( {{{ \boldsymbol{\Phi} }}} \right) &=\mathrm{Tr}\left(  {{{ \boldsymbol{\Phi} }}}^{-1}\frac{d {{{ \boldsymbol{\Phi} }}} }{dz}\right)\label{term3},
	\end{align}
	{\color{black}where ${\boldsymbol{\Delta}}$ is denoted as ${\boldsymbol{\Delta}} = {\boldsymbol{\Psi}} - \overline{\mb{G}}^\dagger {\widetilde{\boldsymbol\Psi}}^{-1} \overline{\mb{G}}^\dagger$.} Therefore  $\mc{G}_{\mb{C}}$ and $	\mc{G}_{\mb{D}}$ can be expressed as
	\begin{align}
	\mc{G}_{\mb{C}}&\!= \left(\boldsymbol{\Delta} - \widetilde{{\boldsymbol\Phi}}^{-1} \right)^{-1}, \\
	\mc{G}_{\mb{D}} &= \left(\widetilde{\boldsymbol{\Phi}}(z) - \boldsymbol{\Delta}^{-1} \right)^{-1}.
	\end{align}

	By applying \eqref{der3}, \eqref{der4} and the Woodbury identity, the forth and the fifth term  can be expressed as 
	\begin{align}
	&\frac{d}{dz}\mtm{Tr}\left( \mc{G}_{\widetilde{\mb{C}}} \left(-z\mb{I}_{LNr}- \widetilde{\boldsymbol{\Psi}} \right) \right) \notag \\
	&=\mathrm{Tr}\left( {	\mc{G}_{\widetilde{\mb{C}}}\frac{d\left( -z\mb{I}_{LNr}- \widetilde{\boldsymbol{\Psi}}\right)}{dz}}\right)+\mathrm{Tr}\left( { \mathrm {diag}\{\widetilde{\eta}_1\left( \mc{G}_{\mb{C}} \right)     ,...,\widetilde{\eta}_L\left( \mc{G}_{\mb{C}} \right) \}\notag	\frac{d\mc{G}_{\widetilde{\mb{C}}}}{dz}} \right)\\\notag
	&= 	\mathrm{Tr}\left( {	\mc{G}_{\widetilde{\mb{C}}}\frac{d\left( -z\mb{I}_{LNr}- \widetilde{\boldsymbol{\Psi}}\right)}{dz}}\right) + \mathrm{Tr}\left( \mc{G}_{\mb{C}}\frac{d{ \left(  \sum\limits_{l=1}^{L}{{{\eta }_{l}}\left( \mc{G}_{{\widetilde{\mb{C}}}_l}  \right)}  \right) } }{dz}   \right)   \\\notag &=-\mathrm{Tr}\left(\mc{G}_{\widetilde{\mb{C}}}\right) - \mathrm{Tr}\left( \mc{G}_{\widetilde{\mb{C}}} \frac{d\widetilde{\boldsymbol{\Psi}} }{dz}\right)-\mathrm{Tr}\left( \mc{G}_{\mb{C}}\frac{d{\boldsymbol{\Psi}} }{dz}   \right)\\\notag
	&=-\mathrm{Tr}\left(\mc{G}_{\widetilde{\mb{C}}}\right)-\mathrm{Tr}\left( \left(\widetilde{\boldsymbol{\Psi}}(z)^{-1} + \widetilde{\boldsymbol{\Psi}}(z)^{-1}\overline{\mb{G}}\mc{G}_{\mb{C}} \overline{\mb{G}}^\dagger\widetilde{\boldsymbol{\Psi}}(z)^{-1} \right) \frac{d\widetilde{\boldsymbol{\Psi}} }{dz}\right)- \mathrm{Tr}\left( \mc{G}_{\mb{C}}\frac{d{\boldsymbol{\Psi}} }{dz}   \right) \\
	&= -\mathrm{Tr}\left(\mc{G}_{\widetilde{\mb{C}}}\right)-	\mathrm{Tr}\left( {	\widetilde{\boldsymbol{\Psi}}^{-1}\frac{d\widetilde{\boldsymbol{\Psi}} }{dz}}\right) + \mathrm{Tr}\left( {\mc{G}_{\mb{C}}	\frac{d\left({ \overline{\mb{G}}^\dagger \widetilde{\boldsymbol{\Psi}}   } \overline{\mb{G}}\right) }{dz}}\right) -\mathrm{Tr}\left( \mc{G}_{\mb{C}}\frac{d{{\boldsymbol\Psi}} }{dz}   \right), \label{term4}
	\end{align}
	
	\begin{align}
	& \frac{d\mathrm{Tr}\left( \mc{G}_{\widetilde{\mb{D}}}\zeta  \right)}{dz} \notag \\
	&=\mathrm{Tr}\left( \mc{G}_{\widetilde{\mb{D}}}\frac{d\zeta }{dz} \right)+\mathrm{Tr}\left( \zeta \frac{d\mc{G}_{\widetilde{\mb{D}}}}{dz} \right)  \notag\\
	&=\mathrm{Tr}\left(\mc{G}_{\widetilde{\mb{D}}}\frac{d\left( \mb{I}-\boldsymbol{\Phi}  \right)}{dz} \right)+\mathrm{Tr}\left( \mc{G}_{{\mb{D}}}\frac{d\tilde{\zeta }}{dz} \right) \notag \\ 
	& =-\mathrm{Tr}\left( \mc{G}_{\widetilde{\mb{D}}}\frac{d\boldsymbol{\Phi} }{dz} \right)-\mathrm{Tr}\left( \mc{G}_{{\mb{D}}}\frac{d\widetilde{\boldsymbol\Phi }}{dz} \right) \notag \\
	&=  -\mathrm{Tr}\left(  {{{ \boldsymbol{\Phi} }}}^{-1}\frac{d {{{ \boldsymbol{\Phi} }}} }{dz}\right) - \mathrm{Tr}\left( \mc{G}_{{\mb{D}}}\frac{d\widetilde{\boldsymbol\Phi }}{dz} \right).\label{term5}		 
	\end{align}
	
	Similarly,  we can get the following equations
	\begin{align}
	& \text{Tr}\left( {{\left( \boldsymbol\Delta -{{ {\widetilde{\boldsymbol\Phi }} }^{-1}} \right)}^{-1}}\frac{d\left( -\boldsymbol\Delta  \right)}{dz} \right) \notag \\
	&=\text{Tr}\left( {{\left( \boldsymbol\Delta -{{ {\widetilde{\boldsymbol\Phi }} }^{-1}} \right)}^{-1}}\frac{d\left({ \overline{\mb{G}}^\dagger \widetilde{\boldsymbol{\Psi}}   } \overline{\mb{G}} - \boldsymbol{\Psi} \right)}{dz} \right) \notag\\  
	& =\text{Tr}\left( {\mc{G}_\mb{C}}\frac{d\left({ \overline{\mb{G}}^\dagger \widetilde{\boldsymbol{\Psi}}   } \overline{\mb{G}}\right)}{dz} \right)-\text{Tr}\left( {\mc{G}_\mb{C}}\frac{d \boldsymbol \Psi }{dz} \right) -\text{Tr}\left( \mc{{G}}_\mb{D}\frac{d\widetilde{\boldsymbol\Phi }}{dz} \right)+\text{Tr}\left( {{\left( \widetilde{\boldsymbol\Phi }-{{\boldsymbol\Delta }^{-1}} \right)}^{-1}}\frac{d\widetilde{\boldsymbol\Phi }}{dz} \right),\\ 
	&\text{Tr}\left( {{\left( \widetilde{\boldsymbol\Phi }-{{\boldsymbol\Delta }^{-1}} \right)}^{-1}}\frac{d\widetilde{\boldsymbol\Phi }}{dz} \right) \notag \\
	&=\text{Tr}\left( {{{\widetilde{\boldsymbol\Phi }}}^{-1}}\frac{d\widetilde{\boldsymbol\Phi }}{dz} \right)-\text{Tr}\left( {{\left( \boldsymbol\Delta -{{ {\widetilde{\boldsymbol\Phi }} }^{-1}} \right)}^{-1}}\frac{d{{ {\widetilde{\boldsymbol\Phi }} }^{-1}}}{dz} \right).
	\end{align}
	Then \eqref{term2_1} can be expressed as
	\begin{align}
	&\frac{d}{dz}\log \det\left(  - {{{{\widetilde {\boldsymbol\Phi} }}}^{ - 1}} + {\boldsymbol\Delta} \right) \notag \\
	&=-\text{Tr}\left( {{\left( \boldsymbol\Delta -{{ {\widetilde{\boldsymbol\Phi }} }^{-1}} \right)}^{-1}}\frac{d\left( -\boldsymbol\Delta  \right)}{dz} \right) -\text{Tr}\left( {{\left( \boldsymbol\Delta -{{ {\widetilde{\boldsymbol\Phi }} }^{-1}} \right)}^{-1}}\frac{d{{ {\widetilde{\boldsymbol\Phi }} }^{-1}}}{dz} \right) ,\notag\\
	& =-\text{Tr}\left( {\mc{G}_\mb{C}}\frac{\left({ \overline{\mb{G}}^\dagger \widetilde{\boldsymbol{\Psi}}   } \overline{\mb{G}}\right)}{dz} \right)+\text{Tr}\left( {\mc{G}_\mb{C}}\frac{d\boldsymbol\Psi }{dz} \right) +\text{Tr}\left( {\mc{G}_\mb{D}}\frac{d\widetilde{\boldsymbol\Phi }}{dz} \right)-\text{Tr}\left( {{{\widetilde{\boldsymbol\Phi }}}^{-1}}\frac{d\widetilde{\boldsymbol\Phi }}{dz} \right).\label{term6}
	\end{align}
	Combining \eqref{term1}, \eqref{term2}, \eqref{term3}, \eqref{term4}, \eqref{term5}, \eqref{term6}, we have
	\begin{align}
	\frac{d}{dz}\mc{V}_{\mb{B}_1}(z)&= -\frac{1}{z}-\frac{1}{LN_r}\mathrm{Tr}\left(\mc{G}_{\widetilde{\mb{C}}}  \right) =-\frac{1}{z}-\mc{G}_{\mb{B}_1}(-z),
	\end{align}
	and the proof of Proposition \ref{prop_shano1} is completed.

\end{document}